\documentclass[reqno,11pt]{article} % font size = 12
\usepackage{amsmath,amsthm,amssymb,url,verbatim,color,mathtools,bbm,xcolor}
\usepackage{stmaryrd}
\newcommand{\remove}[1]{}
\newtheorem{theorem}{Theorem}[section]

\newtheorem{lemma}[theorem]{Lemma}

\newcommand{\xor}{\oplus}

\providecommand{\cX}{\mathcal{X}}
\providecommand{\cY}{\mathcal{Y}}
\DeclareMathOperator{\rk}{rk}

\providecommand{\defeq}{\vcentcolon=}
\renewcommand{\defeq}{\vcentcolon=}

\title{The Unbounded-Error Communication Complexity of symmetric XOR functions}
\author{
Hamed Hatami
\thanks{Supported by an NSERC grant.}
\\
School of Computer Science\\
McGill University, Montreal\\
\texttt{hatami@cs.mcgill.ca}
\and
Yingjie Qian
\\
Department of Mathematics and Statistics\\
McGill University, Montreal\\
\texttt{yingjie.qian@mail.mcgill.ca}
}

\begin{document}

\maketitle
\begin{abstract}
Settling a conjecture of Shi and Zhang~\cite{MR2553112}, we determine the unbounded-error communication complexity of the symmetric XOR functions up to a poly-logarithmic factor. Our proof is by a simple reduction to an earlier result of Sherstov regarding  the symmetric AND functions. 

\end{abstract}

\section{Introduction}

The \emph{unbounded-error} model, defined by Paturi and Simon~\cite{MR864082}, is a deep and elegant communication model with several applications to circuit complexity and learning theory. The setting is the same as Yao's standard two-party communication model~\cite{Yao:1979} with the players having access only to \emph{private} random coins. The goal of the players is to agree on an output that is better than a completely random guess. More formally, given a communication problem $f: \cX \times \cY \to \{-1,1\}$, Alice and Bob receive $x\in \cX$ and $y\in \cY$ respectively. They communicate according to an agreed upon communication protocol $\pi$ that has access to private randomness, and their goal is to achieve $\Pr[\pi(x,y)=f(x,y)]>\frac{1}{2}$ for all $(x,y) \in \cX \times \cY$. The cost of $\pi$ is the worst-case number of bits exchanged on any input $(x,y)$, and the unbounded-error communication complexity of $f$, denoted by $U(f)$, is the least cost of such a protocol.

The unbounded-error model is one of the most powerful communication models as it is stronger than any of the usual deterministic, nondeterministic, randomized, or quantum communication models. As a result, obtaining lower bounds for  unbounded-error communication complexity is very desirable as such results  also imply lower bounds for other notions of communication complexity. However, often the unbounded-error complexity is exponentially smaller than the complexity of the function in these other models.

A beautiful fact about unbounded-error communication complexity, discovered by  Paturi and Simon~\cite{MR864082}, is that $U(f)$ has an elegant characterization in terms of a matrix quantity called the sign-rank. More precisely $U(f)=\log_{2} \rk_\pm(f) + O(1)$ where the sign-rank of $f$, denoted by $\rk_{\pm}(f)$, is the smallest rank of an $\cX \times \cY$ matrix whose every entry is a non-zero real with the same sign as $f$. Despite this elegant characterization, due to its inherent power,  no nontrivial lower bounds were known for any explicit function  in this model until the breakthrough work of Forster~\cite{MR1964645}, who proved a strong lower bound for the inner product function, and  more generally, for any function whose communication matrix has low spectral norm.  Forster's deep inequality has been since the main tool in establishing lower bounds in this model, and all the major subsequent works (e.g.~\cite{MR2886100,MR2592035}) have been essentially based on variations of Forster's inequality and combining it with other techniques such as Sherstov's pattern matrix method~\cite{MR2582888}.

The focus of this note is  on symmetric XOR functions, i.e., functions $f^\xor:\{0,1\}^n\times\{0,1\}^n\to\{-1,1\}$ of the form $f^\xor(x,y)=D(\sum_{i=1}^n x_i \xor y_i)$, where $D:\{0,1,\dots,n\}\to\{-1,1\}$ is a given predicate, and $x_i \xor y_i$ stands for the \emph{exclusive or} of the two bits $x_i$ and $y_i$. In~\cite{MR2553112} Shi and Zhang characterized the bounded-error randomized and quantum communication complexities of such functions up to poly-logarithmic factors. They observed that the bounded-error randomized and quantum case can be reduced to the case of the symmetric AND function $f^\wedge(x,y)=D(\sum_{i=1}^n x_iy_i)$, a problem which was solved earlier in an important paper of Razborov~\cite{MR1957920}.

The problem of determining the unbounded-error communication complexity of the functions of the form $f^\wedge(x,y)$ has been resolved as well. Indeed  Sherstov combined ideas from Razborov's paper~\cite{MR1957920} with his  pattern matrix method~\cite{MR2582888}, and a generalization of Forster's theorem~\cite{MR1964645} to prove that the randomized communication complexity of $f^\wedge(x,y)$ is essentially equal to the number of sign changes in $D$, i.e. $|\{i : D(i) \neq D(i + 1)\}|$. Shi and Zhang~\cite{MR2553112} conjectured that similarly the unbounded-error complexity of $f^\xor=D(\sum_{i=1}^n x_i \xor y_i)$ is essentially equal to $|\{i : D(i) \neq D(i + 2)\}|$. However, they speculated that Sherstov's approach cannot  be applied to this problem. 

Recently Chattopadhyay and Mande~\cite{CN17} made partial progress towards resolving this conjecture using a direct approach based on Fourier analysis.  In this note we settle the conjecture by showing that, contrary to the belief of Shi and Zhang,  it can also be deduced from Sherstov's result via a simple reduction.\footnote{After distributing a preliminary version of this paper, we learned that Anil Ada, Omar Fawzi, and Raghav Kulkarni were preparing an article~\cite{AOR17}  in which they have independently proved this conjecture as well as several other  results. Indeed the first author has  learned about the connection between the symmetric AND functions and the symmetric XOR functions from a correspondence with Omar Fawzi a few years ago  regarding a different problem.}

\subsection{Notation}
For a natural number $n$, we use $[n]$ to denote the set $\{1,\ldots,n\}$. Given a predicate $D:\{0,1,\dots,n\}\to\{-1,1\}$, the functions $f^\xor_D, f^\wedge_D:\{0,1\}^n\times\{0,1\}^n\to\{-1,1\}$ are defined, respectively, as $f_D^\xor(x,y)=D(\sum_{i=1}^n x_i \xor y_i)$, and  $f_D^\wedge(x,y)=D(\sum_{i=1}^n x_iy_i)$.

We denote the \emph{Hamming weight} of a vector $z \in \{0,1\}^n$ by $|z|= \sum_{i=1}^n z_i$. Let $x \xor y$ and $x \wedge y$ respectively denote the bitwise XOR and the bitwise AND of the two vectors $x,y \in \{0,1\}^n$. In this notation, we have $f^\xor_D=D(|x \xor y|)$ and $f^\wedge_D(x,y)=D(|x \wedge y|)$. We also write $f|_{r,t}$ if the inputs $x,y$ are restricted to satisfy $|x|=r$ and $|y|=t$.

The degree $\deg(D)$ of a given predicate $D:\{0,1\ldots,n\} \to\{-1,1\}$ is the number of times $D$ changes value in $\{0,1\ldots,n\}$. In other words, $\deg(D)=|\{i : D(i) \neq D(i + 1)\}|$. It is not difficult to show that $\deg(D)$ is the least degree of a real univariate polynomial $p$ such that $p(i)$ has the same sign as $D(i)$ for every $i$, whence the name. We similarly define $\deg_2(D)= |\{i : D(i) \neq D(i + 2)\}|$.

\section{Main result}

In this section we state and prove our main result, establishing the conjecture of Shi and Zhang~\cite{MR2553112}.

\begin{theorem}[Main Result]
\label{thm:main}
Let $D:\{0,1,\dots,n\}\to\{-1,1\}$ be a given predicate, and let $M=\deg_2(D)$. We have
$$\Theta\left(\frac{M}{\log^5(n)}\right) \le U(f^{\xor}_D) \le \Theta(M \log n).$$
\end{theorem}

The upper bound is easy and was probably known to~\cite{MR2553112}. We prove the lower bound by reducing it to the following result of  Sherstov~\cite{MR2886100} on unbounded-error communication complexity of the symmetric AND functions.

%, however we will present its proof for completeness. 
\begin{theorem}[{\cite{MR2886100}}]
\label{thm:sherstov}
Let $D:\{0,1,\dots,n\}\to\{-1,1\}$ be a given predicate and $K=\deg(D)$. We have,
$$\Theta\Big(\frac{K}{\log^5(n)}\Big) \le U(f^{\wedge}_D) \le \Theta(K \log n).$$
\end{theorem}

\subsection{Proof of Theorem~\ref{thm:main}, lower bound}
We start with a simple observation that allows us to ``reverse'' the predicate  if necessary.

\begin{lemma}
\label{lem:flip}
Consider $D:\{0,1,\dots,n\}\to\{-1,1\}$, and let the reverse of $D$ be the function $\overleftarrow{D}:\{0,1,\dots,n\}\to\{-1,1\}$ defined as $\overleftarrow{D}(i)=D(n-i)$ for all $i$. We have
$$U(f^{\xor}_{\overleftarrow{D}})=U(f^{\xor}_D).$$
\end{lemma}
\begin{proof}
Let $\pi$ be a communication protocol in the unbounded-error model for the function $f^{\xor}_D$. Consider the new protocol in which Alice first negates her input by replacing all her $0$'s with $1$'s and all her $1$'s with $0$'s to obtain a new vector $x'$. Then the two players proceed by running $\pi$ on $x'$ and $y$. Since $|x'\xor y|=n-|x\xor y|$, the new protocol is a valid protocol for $f^{\xor}_{\overleftarrow{D}}$, and thus $U(f^{\xor}_{\overleftarrow{D}})\leq U(f^{\xor}_D)$. Reversing $\overleftarrow{D}$, we obtain $U(f^{\xor}_D)\leq U(f^{\xor}_{\overleftarrow{D}})$.
\end{proof}

To avoid rounding issues, we assume that $n$ is a sufficiently large power of two. Indeed any communication protocol that computes $D$ can clearly compute the restriction of $D$ to a given subinterval. Hence in the case where $n$ is not a power of two,  we can restrict to a subinterval whose length is a power of two while ensuring that the $\deg_2$ of the new problem is at least $M/2$. 

Hence, below, $n$ is assumed to be a sufficiently large power of two, and $M=\deg_2(D)$.  Set  $q=\frac{n}{2^5}$ and $r=\frac{n-q}{2}$. By Lemma~\ref{lem:flip}, replacing $D$ with $\overleftarrow{D}$ if necessary, we can assume, without loss of generality, that there are at least $M/2$ indices $0 \le i < n/2$ for which  $D(i) \neq D(i+2)$. Out of those,  either more than half are odd, or at least half are even. Thus by restricting to an interval of length $q$, we conclude that there exists an integer $s \in \{0,q,2q,\ldots,15 q\}$ such that either there are at least $\frac{M}{2^7}$ even values of $i \in [s,s+q)$ for which  $D(i) \neq D(i+2)$ or there are at least $\frac{M}{2^7}$ such odd values. If it is the former case, set  $t=r$, and otherwise set $t=r+1$.

 Define $G:\{0,\ldots,q/2\} \to \{-1,1\}$ as 
$$G(i) \defeq D(r+t-(n-2q-s)-2i)
=
\left\{
\begin{array}{lcl}
 D(s+q-2i) &\qquad &r=t \\
 D(s+q-2i+1) &\qquad &r=t+1
\end{array}\right.
.$$ 
By the above discussion  $K_G=\deg(G) \ge \frac{M}{2^7}$. Thus, by Theorem~\ref{thm:sherstov}, $$U(f^{\wedge}_G)\geq\Theta\Big(\frac{K_G}{\log^5q}\Big)=\Theta\Big(\frac{M}{\log^5n}\Big).$$ 

Let $k=\frac{n-2q-s}{2}$, and note  $G(i)=D(r+t-2(k+i))$. To finish the proof we will embed $f_G^\wedge$ in $f_D^\xor$. Alice and Bob receive two inputs $x',y'\in\{0,1\}^{q/2}$ for the $f_G^\wedge$  function. They create  inputs $x,y \in \{0,1\}^n$ for $f_D^\xor$ satisfying $|x|=r$ and $|y|=t$ in the following manner. Alice creates $x$ by adding  $r-|x'|$ number of $1$'s to $x'$ and $n-\frac{q}{2}-(r-|x'|)$ number of $0$'s and Bob creates $y$ by adding $t-|y'|$ number of $1$'s and $n-\frac{q}{2}-(t-|y'|)$ number of $0$'s. Moreover, Alice and Bob ensure that the added $1$'s match in $k$ positions. This is always possible as
$$k+\frac{q}{2}\leq r  \qquad \mbox{and}\qquad r+t-k\leq n-\frac{q}{2}.$$

Note that $|x\wedge y|=k+|x'\wedge y'|$ from our construction, and $|x\xor y|=r+t-2|x\wedge y|$ holds whenever $|x|=r$ and $|y|=t$. So $|x\xor y|=r+t-2(k+|x'\wedge y'|)$, and $G(|x'\wedge y'|)=D(r+t-2(k+|x'\wedge y'|))=D(|x\xor y|)$.  Hence, 
$$U(f^{\xor}_D)\geq U(f^{\wedge}_G)\geq\Theta\Big(\frac{M}{\log^5n}\Big).$$

\subsection{Proof of Theorem~\ref{thm:main}, upper bound}
As we mentioned earlier, the upper bound is straightforward. Recall that the Fourier expansion of $p:\{0,1\}^n \to \mathbb{R}$ is the unique expansion $p(x)=\sum_{S \subseteq [n]} \widehat{p}(S) \chi_S(x)$, where  $\chi_{S}(x)\defeq (-1)^{\sum_{i \in S} x_i}$ are the Fourier characters. The real numbers $\widehat{p}(S)$ are called Fourier coefficients.

Consider a function $p:\{0,1\}^n \to \mathbb{R}$, and define $p^\oplus:\{0,1\}^n \times \{0,1\}^n \to \mathbb{R}$  as $p^\oplus(x,y)\defeq p(x \oplus y)$.  It is well-known and easy to see that the rank of the $2^n \times 2^n$ matrix with entries $p^\oplus(x,y)$  is equal to the number of non-zero Fourier coefficients of $p$. Let $D:\{0,1,\dots,n\}\to\{-1,1\}$ be a given predicate, and let $M=\deg_2(D)$. We will prove the upper bound by constructing a function $p:\{0,1\}^n \to \{-1,1\}$ with at most $4n^M$ non-zero Fourier coefficients such that  $p(x) D(\sum x_i)>0$ for all $x \in \{0,1\}^n$. Then $p^\oplus$  sign-represents $f^\oplus_D$ and yields the bound $\rk_\pm(f^\oplus_D) \le 4 n^M$ implying the desired upper bound .

Consider the two functions $A,B:\{0,\ldots,\lfloor n/2 \rfloor \} \to \{-1,1\}$ defined as $A(y) \defeq D(2y)$ and $B(y) \defeq D(2y+1)$, where we set $D(n+1)\defeq D(n-1)$ in the case of $2y+1=n+1$.   Obviously $\deg(A),\deg(B) \le M$, and thus there are univariate polynomials $q_1(y)$ and $q_2(y)$ such that $\deg(q_1),\deg(q_2) \le M$ and $q_1(y)A(y)>0$ and $q_2(y)B(y)>0$ for all $y \in \{0,\ldots,\lfloor n/2 \rfloor \}$. 
 
Let
$$p(x) =\frac{1 + \chi_{[n]}(x)}{2} q_1\left(\frac{\sum x_i }{2}\right) +\frac{1 - \chi_{[n]}(x)}{2} q_2\left(\frac{(\sum x_i)-1}{2} \right),$$
and note that if $\sum x_i$ is even, then $p(x) = q_1(y)$ where $2y=\sum x_i$, and  if $\sum x_i$ is odd, then $p(x) = q_2(y)$ where $2y+1=\sum x_i$. In  both cases $p(x)$ has the same sign as $D(\sum x_i)$. On the other hand since $q_1\left(\frac{\sum x_i }{2}\right)$ and $q_2\left(\frac{(\sum x_i)-1}{2} \right)$ are both polynomials of degree at most $M$, their Fourier expansion is supported on sets $S$ of size at most $M$. Multiplying by  $\frac{1 + \chi_{[n]}(x)}{2}$ and $\frac{1 - \chi_{[n]}(x)}{2}$ at most doubles the support of the Fourier expansion. Hence, for sufficiently large $n$, the Fourier expansion of $p$ is supported on at most $4 \sum_{k=0}^M {\lfloor n/2 \rfloor +1  \choose k} \le 4 n^M$ terms as desired.

\bibliographystyle{alpha}
\bibliography{XOR}

\end{document}